\documentclass{amsart}
\usepackage{amssymb,amsmath,amscd, amsthm,stmaryrd,verbatim, mathrsfs}

\newtheorem{Thm}{Theorem}
\newtheorem{theorem}{Theorem}[section]

\newtheorem{Th}[theorem]{Theorem}
\newtheorem{Lem}[theorem]{Lemma}

\newtheorem{Prop}[theorem]{Proposition}

\theoremstyle{definition}
\newtheorem{Def}[theorem]{Definition}

\theoremstyle{remark}
\newtheorem{rmk}[theorem]{Remark}

\numberwithin{equation}{section}

\newcommand{\bb}{\mathbb}
\newcommand{\ms}{\mathscr}
\newcommand{\mr}{\mathrm}
\newcommand{\frk}{\mathfrak}

\usepackage{hyperref}
\begin{document}

\title{Generalized MICZ-Kepler Problems and Unitary Highest Weight Modules -- II}
\author{Guowu Meng}

\address{Department of Mathematics, Hong Kong Univ. of Sci. and
Tech., Clear Water Bay, Kowloon, Hong Kong}

\email{mameng@ust.hk}


\date{April 23, 2007}



\begin{abstract}
For each integer $n\ge 2$, we demonstrate that a $2n$-dimensional
generalized MICZ-Kepler problem has an $\widetilde{\mr{Spin}}(2,
2n+1)$ dynamical symmetry which extends the manifest $\mr{Spin}(2n)$
symmetry. The Hilbert space of bound states is shown to form a
unitary highest weight $\widetilde{\mr{Spin}}(2, 2n+1)$-module which
occurs at the first reduction point in the Enright-Howe-Wallach
classification diagram for the unitary highest weight modules. As a
byproduct, we get a simple geometric realization for such a unitary
highest weight $\widetilde{\mr{Spin}}(2, 2n+1)$-module.
\end{abstract}

\maketitle

\section {Introduction}
The Kepler problem is a physics problem in dimension three about two
bodies which attract each other by a force proportional to the
inverse square of their distance. The MICZ-Kepler problems,
discovered in the late 60s by McIntosh and Cisneros \cite{MC70} and
independently by Zwanziger \cite{Z68}, are natural cousins of the
Kepler problem. It was shown by Barut and Bornzin \cite{Barut71}
that the MICZ-Kepler problems all have a large dynamical symmetry
group---$\mr{Spin}(2, 4)$.

The $5$-dimensional analogues of the MICZ-Kepler problems were found
by Iwai \cite{Iwai90} and their dynamical symmetries were studied by
Pletyukhov and Tolkachev \cite{Pletyukhov99}. About two years ago,
the $D$-dimensional ($D\ge 3$) analogues of the MICZ-Kepler
problems, which extends both the MICZ-Kepler problems and Iwai's
$5$-dimensional analogues, were found by this author \cite{meng05}.
As in Ref. \cite{MZ07} we shall refer to the MICZ-Kepler problems
and their higher dimensional analogues as the \emph{generalized
MICZ-Kepler problems}. It is then natural to extend the study of the
dynamical symmetry to all generalized MICZ-Kepler problems. Indeed,
when the dimension $D\ge 3$ is an odd integer, this study has been
carried out recently by Zhang and this author \cite{MZ07}; in fact,
more refined results, which are of interest to representation
theorists, were obtained there.

The purpose of the present paper is to carry out this study in the
case when $D\ge 3$ is an even integer. Note that, when $D$ is even,
the magnetic charge must be either $0$ or $1/2$. We shall show that
for each integer $n\ge 2$, a $2n$-dimensional generalized
MICZ-Kepler problem always has an $\widetilde{\mr{Spin}}(2,
2n+1)$\footnote{It is a $4$-fold cover of $\mr{SO}_0(2, 2n+1)$ ---
the identity component of $\mr{SO}(2, 2n+1)$. By definition, it is
characterized by the homomorphism $\pi_1(\mr{SO}_0(2, 2n+1))=\bb
Z\oplus \bb Z_2\to \bb Z_2\oplus \bb Z_2$ which maps $(a, b)$ to
$(\bar a, b)$. Here $\bar a$ is the congruence class of $a$ modulo
$2$.} dynamical symmetry, i.e., \emph{its Hilbert space of bound
states forms an irreducible unitary module for
$\widetilde{\mr{Spin}}(2, 2n+1)$}. In fact, we shall show that the
Hilbert space of bound states forms a unitary highest weight module
for $\widetilde{\mr{Spin}}(2, 2n+1)$; more precisely, we shall
establish the following result:

\begin{Thm}\label{main} Assume $n>1$ is an integer and $\mu=0$ or
$1/2$. Let ${\ms H}(\mu)$ be the Hilbert space of bound states for
the $2n$-dimensional generalized MICZ-Kepler problem with magnetic
charge $\mu$, and $l_\mu=l+\mu+n-3/2$ for any integer $l\ge 0$.

1) The manifest unitary action of $\mr{Spin}(2n)$ on ${\ms H}(\mu)$
extends to a natural unitary action of $\mr{Spin}(2, 1)\times
\mr{Spin}(2n)$ under which ${\ms H}(\mu)$ decomposes as follows:
\begin{eqnarray}
{\ms H}(\mu) =\left\{ \begin{array}{ll}
 \hat \bigoplus_{l=0}^\infty
\left({\mathcal D}^-_{2l_\mu+2}\otimes (D_l^+\oplus D_l^-)\right) & \mbox{if $\mu=1/2$}\\
\\
\hat \bigoplus_{l=0}^\infty \left({\mathcal D}^-_{2l_\mu+2}\otimes
D_l^0\right) & \mbox{if $\mu=0$}
\end{array}\right.
\end{eqnarray} where $D_l^\pm$ is the irreducible module of $\mr{Spin}(2n)$
with highest weight $$(l+1/2, 1/2, \cdots,1/2, \pm{1/2}),$$ $D_l^0$
is the irreducible module of $\mr{Spin}(2n)$ with highest weight
$(l, 0, \cdots,0)$,  and ${\mathcal D}^-_{2l_\mu+2}$ is the
anti-holomorphic discrete series representation of $\mr{Spin}(2, 1)$
with highest weight $-l_\mu-1$.

2) There is a natural unitary action of $\widetilde{\mr{Spin}}(2,
2n+1)$ on $ {\ms H}(\mu)$ which extends the preceding unitary action
of $\mr{Spin}(2, 1)\times \mr{Spin}(2n)$. In fact, ${\ms H}(\mu)$ is
the unitary highest weight module of $\widetilde{\mr{Spin}}(2,
2n+1)$ with highest weight $\left(-(n+\mu-1/2),\mu, \cdots,
\mu\right)$; consequently, it occurs at the first reduction point of
the Enright-Howe-Wallach classification diagram\footnote{Page 101,
Ref. \cite{EHW82}. While there is a unique reduction point when
$\mu=1/2$, there are two reduction points when $\mu=0$ with the
second reduction point being the trivial representation. See also
Refs. \cite{Jakobsen81a,Jakobsen81b}.} for the unitary highest
weight modules, so it is a non-discrete series representation.

3) As a representation of the maximal compact subgroup
$\mr{Spin}(2)\times \mr{Spin}(2n+1)$,
\begin{eqnarray}
{\ms H}(\mu) = \hat \bigoplus_{l=0}^\infty \left(D(-l_\mu-1)\otimes
D^l\right)
\end{eqnarray} where $D^l$ is the irreducible module of $\mr{Spin}(2n+1)$
with highest weight $(l+\mu, \mu, \cdots,  \mu)$ and $D(-l_\mu-1)$
is the irreducible module of $\mr{Spin}(2)$ with weight $-l_\mu-1$.

\end{Thm}
We would like to point out that the unitary highest weight
representation of $\widetilde{\mr{Spin}}(2, 2n+1)$ with highest
weight $\left(-(n+\mu-1/2),\mu, \cdots, \mu\right)$ exists if and
only if when $\mu=0$ or $1/2$, see pages 127-128 in Ref.
\cite{EHW82} and the paragraph containing Eq. (2.7) in Ref.
\cite{EHW82}. Therefore, we have another representation theoretic
evidence for the nonexistence of the even-dimensional generalized
MICZ-Kepler problem with magnetic charge higher than $1/2$.

Readers who wish to have a quick geometric description of the
aforementioned unitary highest weight module of
$\widetilde{\mr{Spin}}(2, 2n+1)$ may consult the appendix. Readers
who wish to know more details about the classification
\cite{Jakobsen81a, Jakobsen81b, EHW82} of unitary highest weight
modules may start with a fairly readable account from Ref.
\cite{EHW82}. Note that there is no general classification result
for the family of unitary modules of real non-compact simple Lie
groups, and the subfamily of unitary \emph{highest weight} modules
is special enough so that such a nice classification result can
possibly exist. The first reduction point picked up by the ``Nature"
from the Enright-Howe-Wallach classification diagram is even more
special because it belongs to an even more special subfamily called
Wallach set.

\vskip 10pt In section \ref{GMIC}, we give a quick review of the
generalized MICZ-Kepler problems in even dimensions. For the
computational purpose in the subsequent section, we quickly review
the gauge potential\footnote{It is $\sqrt {-1}$ times the local {\em
connection one-form}. } for the background gauge field (i.e.,
\emph{connection}) under a particular local gauge (i.e., {\em bundle
trivialization}), and then quote from Ref. \cite{meng05} some key
identities satisfied by the gauge potential. In section \ref{HDS},
we introduce the dynamical symmetry operators and show that they
satisfy the commutation relations for the generators\footnote{Here
we adopt the practice in physics: the Lie algebra generators act as
hermitian operators in all unitary representations.} of
$\frk{so}_0(2, 2n+1)$. We also show that these dynamical symmetry
operators satisfy a set of quadratic relations \footnote{This set of
quadratic relations will be shown \cite{meng07} to algebraically
characterize the unitary highest weight $\widetilde{\mr{Spin}}(2,
2n+1)$-modules stated in Theorem \ref{main} above.}. These results
are stated as Theorem \ref{keyr}. In section \ref{Rep}, we start
with a preliminary discussion of the representation problem and
point out the need of ``twisting". Then we gave a review of the
(bound) energy eigenspaces (i.e., \emph{eigenspaces of the
harmiltonian} viewed as a hermitian operator on the physical Hilbert
space) and finally introduce the notion of ``twisted'' energy
eigenspaces which is soon shown to be the space of $L^2$-sections of
a canonical hermitian bundle. In the last section, we solve the
representation problem by proving two propositions from which
Theorem \ref{main} follows quickly. In the appendix, each of the
unitary highest weight representation of $\mr{Spin}(2, 2n+1)$
encountered here is geometrically realized as the space of all $L^2$
sections of a canonical hermitian bundle. Via communications with
Profs. R. Howe and N. Wallach, we learned that these representations
can be imbedded into the kernel of certain canonical differential
operators, see Refs. \cite{Tan-Howe, KO03} for the case $\mu=0$ and
Ref. \cite{EW97} for the general case.

\section{Review of generalized MICZ-Kepler problems}\label{GMIC}
From the physics point of view, a MICZ-Kepler problem is a
generalization of the Kepler problem by adding a suitable background
magnetic field, while at the same time making an appropriate
adjustment to the scalar Coulomb potential so that the problem is
still integrable. The configuration space is the punctured 3D
Euclidean space, and the background magnetic field is a Dirac
monopole. To be more precise, the (dimensionless) hamiltonian of a
MICZ-Kepler problem with magnetic charge $\mu $ is
\begin{eqnarray}H
 = -\frac{1}{2}{\Delta}_{\mathcal A} + \frac{\mu^2}{2r^2}
 - \frac{1}{r}\;.
\end{eqnarray}
Here $\Delta_{\mathcal A}$ is the Laplace operator twisted by the
gauge potential $\mathcal A$ of a Dirac monopole under a particular
gauge, and $\mu$ is the magnetic charge of the Dirac monopole, which
must be a half integer.

To extend the MICZ-Kepler problems beyond dimension three, one needs
a suitable generalization of the Dirac monopoles. Fortunately this
problem was solved in Refs. \cite{ meng04, Cotaescu05, meng05}. We
review the work here.

\subsection{Generalized MICZ-Kepler problems}

Let $D\ge 3$ be an integer, $\bb R^{D}_*$ be the punctured
$D$-space, i.e., $\bb R^D$ with the origin removed. Let $ds^2$ be
the cylindrical metric on $\bb R^D_*$. Then $(\bb R_*^D, ds^2)$ is
the product of the straight line $\bb R$ with the round sphere ${\mr
S}^{D-1}$. Since we are interested in the even dimensional
generalized MICZ-Kepler problems only in this paper, we assume $D$
is even.

Let $\mathcal S$ be the spinor bundle of $(\bb R^D_*, ds^2)$, then
$\mathcal S$ corresponds to the fundamental spin representations
${\bf s}$ of $\frk{so}_0(D-1)$ (The Lie algebra of $\mr{SO}(D-1)$).
Note that this spinor bundle is endowed with a natural $\mr{SO}(D)$
invariant connection
--- the Levi-Civita spin connection of $(\bb R^D_*, ds^2)$. As a
result, the Young product of $I$ copies of these bundles, denoted by
${\mathcal S}^I$, are also equipped with natural $\mr{SO}(D)$
invariant connections. (By convention, ${\mathcal S}^{0}$ is the
product complex line bundle with the product connection.) The
corresponding representation of $\frk{so}_0(D-1)$ will be denoted by
${\bf s}^{2\mu}$.

\begin{Def} Let $n\ge 2$ be an integer and $\mu= 0$ or $1/2$.
The $2n$-dimensional generalized MICZ-Kepler problem with magnetic
charge $\mu$ is defined to be the quantum mechanical system on $\bb
R^{2n}_*$ for which the wave-functions are sections of ${\mathcal
S}^{2\mu}$, and the hamiltonian is
\begin{eqnarray}
H
 = -\frac{1}{2}\Delta_\mu + \frac{(n-1)\mu}{2r^2}
 - \frac{1}{r}
\end{eqnarray}
where $\Delta_\mu$ is the Laplace operator twisted by ${\mathcal
S}^{2\mu}$.
\end{Def}

Upon choosing a local gauge, the background gauge field (i.e., the
natural connection on ${\mathcal S}^{2\mu}$) can be represented by a
gauge potential ${\mathcal A}_\alpha$ in an explicit form; then
$\Delta_\mu$ can be represented explicitly by
$\sum_\alpha(\partial_\alpha +i {\mathcal A}_\alpha)^2$. Since the
gauge potential is of crucial importance, we review some of its
properties in the next subsection.

\subsection {Basic identities for the gauge potential}\label{GDM}

We write $\vec r = (x_1, x_2, \ldots, x_{D-1}, x_D)$ for a point in
$\bb R^{D}$ and $r$ for the length of $\vec r$. The small Greek
letters $\mu$, $\nu$, etc run from $1$ to $D$ and the lower case
Latin letters $a$, $b$ etc run from $1$ to $D-1$. We use the
Einstein convention that repeated indices are always summed over.

Under a suitable choice of local gauge on ${\bb R}^{D}$ with the
negative $D$-th axis removed, the gauge field can be represented by
the following gauge potential:
\begin{eqnarray}\label{mnple}
{\mathcal A}_D=0,\hskip 20 pt {\mathcal A}_b=-{1\over
r(r+x_D)}x_a\gamma_{ab}
\end{eqnarray}
where $\gamma_{ab}={i\over 4}[\gamma_a,\gamma_b]$ with $\gamma_a$
being the ``gamma matrix" for physicists.  Note that $\gamma_a=ie_a$
with $e_a$ being the element in the Clifford algebra that
corresponds to the $a$-th standard coordinate vector of $\bb
R^{D-1}$.

The field strength of $\mathcal A_\alpha$ is then given by
\begin{eqnarray}
\left\{\begin{array}{rcl} F_{Db}& = & {1\over
r^3}x_a\gamma_{ab},\cr\\ F_{ab}& = & -{2\gamma_{ab}\over r(r+x_D)}+
{1\over r^2(r+x_D)^2}\cdot\cr & &\left((2+{x_D\over r})x_c(x_a
\gamma_{cb}-x_b \gamma_{ca}) +ix_d x_c[\gamma_{d a},\gamma_{c
b}]\right)\end{array}\right.
\end{eqnarray}

Here are some identities from Ref. \cite{meng05} that our later
computations will crucially depend on:
\begin{Lem}\label{lemma}
Let $\mathcal A_\alpha$ be the gauge potential defined by Eq.
(\ref{mnple}) and let $F_{\alpha\beta}$ be its field strength.

1) The following identities are valid in any representation of
$\frk{so}_0(D-1)$:
\begin{eqnarray}\label{Id}
F_{\mu\nu}F^{\mu\nu}=\frac{2}{r^4}c_2,\quad
{[\nabla_\kappa,F_{\mu\nu}]}={1\over r^2}\left( x_\mu
F_{\nu\kappa}+x_\nu F_{\kappa \mu}-2x_\kappa F_{\mu\nu} \right), \cr
x_\mu {\mathcal A}_\mu=0,\hskip 20pt x_\mu F_{\mu\nu}=0, \hskip 20pt
[\nabla_\mu, F_{\mu\nu}]=0, \cr r^2[F_{\mu\nu},
F_{\alpha\beta}]+iF_{\mu\beta}\delta_{\alpha\nu}-
iF_{\nu\beta}\delta_{\alpha\mu}+iF_{\alpha\mu}\delta_{\beta\nu}-iF_{\alpha\nu}\delta_{\beta\mu}\cr
={i\over r^2}\left(x_\mu x_\alpha F_{\beta\nu}+x_\mu x_\beta
F_{\nu\alpha}-x_\nu x_\alpha F_{\beta\mu}-x_\nu x_\beta
F_{\mu\alpha} \right),
\end{eqnarray}
where $\nabla_\alpha=\partial_\alpha+i\mathcal A_\alpha$, and
$c_2=c_2[\frk{so}_0(D-1)]={1\over 2}\gamma_{ab}\gamma_{ab}$ is the
(quadratic) Casimir operator of $\frk{so}_0(D-1)$.

2) When $D=2n$, identity
\begin{eqnarray}
r^2F_{\lambda\alpha}F_{\lambda\beta}={n-1\over 2}\left({1\over
r^2}\delta_{\alpha\beta}-{x_\alpha x_\beta\over
r^4}\right)+i(n-{3\over 2})F_{\alpha\beta}
\end{eqnarray} holds in the fundamental spin representation $\bf s$ of $\frk{so}_0(2n-1)$.
\end{Lem}
Remark that ${\mathcal A}_r={\mathcal A}_\theta=0$, where ${\mathcal
A}_r$ and ${\mathcal A}_\theta$ are the $r$ and $\theta$ components
of $\mathcal A$ in the polar coordinate system $(r, \theta,
\theta_1, \cdots,\theta_{D-3},\phi)$ for $\bb R^D_*$ with $\theta$
being the angle between $\vec r$ and the positive $D$-th axis.

\section {The dynamical symmetry}\label{HDS}
For the remainder of this paper, we only consider a fixed
$2n$-dimensional generalized MICZ-Kepler problem. Recall that the
configuration space is $\bb R^D_*$ where $D=2n$. For our
computational purposes, it suffices to work on ${\bb R}^D$ with the
negative $D$-axis removed. Introduce the notations
$\pi_\alpha:=-i\nabla_\alpha$, $c:=(n-1)\mu$. Then $[\pi_\alpha,
\pi_\beta]=-iF_{\alpha\beta}$.

Just as in Ref. \cite{MZ07}, we let
\begin{eqnarray}\label{def1}
\left\{\begin{array}{rcl}\vec \Gamma & := &r\vec \pi, \hskip 10pt X
:=r\pi^2+{c\over r}, \hskip 10pt
 Y := r,\cr\\
 J_{\alpha\beta}  &:=& i[\Gamma_\alpha, \Gamma_\beta],\hskip
 10pt
 \vec Z :=
i[\vec \Gamma, X],\hskip 10pt\vec W :=  i[\vec \Gamma, Y]=\vec r;
\end{array}\right.
\end{eqnarray}
and
\begin{eqnarray}\label{def2}
\left\{\begin{array}{l} \Gamma_{D+1} := {1\over 2}\left(X-Y \right),
\hskip 30pt
 \Gamma_{-1} := {1\over 2}\left(X+Y \right),
 \cr\\
 \vec A := {1\over 2}\left(\vec Z-\vec W \right),  \hskip 10pt
 \vec M := {1\over 2}\left(\vec Z+\vec W \right),\hskip 10pt
 T :=i[\Gamma_{D+1}, \Gamma_{-1}].\end{array}\right.
\end{eqnarray}
Some relatively straightforward but lengthy computations yield
\begin{eqnarray}\left\{\begin{array}{rcl}
J_{\alpha\beta}  &= & x_\alpha\pi_\beta
-x_\beta\pi_\alpha+r^2F_{\alpha\beta},\cr\\ A_\alpha &= &{1\over
2}x_\alpha\pi^2 - \pi_\alpha(\vec r\cdot \vec \pi)+r^2F_{\alpha
\beta}\pi_\beta-{c\over 2r^2}x_\alpha +{i\over
2}(D-3)\pi_\alpha-{1\over 2}x_\alpha,\cr\\ M_\alpha &= &{1\over
2}x_\alpha\pi^2 - \pi_\alpha(\vec r\cdot \vec \pi)+r^2F_{\alpha
\beta}\pi_\beta-{c\over 2r^2}x_\alpha +{i\over
2}(D-3)\pi_\alpha+{1\over 2}x_\alpha,\cr\\ T  &= & \vec r\cdot \vec
\pi-i{D-1\over 2},\cr\\ \Gamma_\alpha &=& r\pi_\alpha,\cr\\
\Gamma_{-1} &= & {1\over 2}\left(r\pi^2+r+{c\over r} \right),\cr\\
\Gamma_{D+1} &= & {1\over 2}\left(r\pi^2-r+{c\over r}
\right).\end{array}\right.\nonumber
\end{eqnarray}

Let the capital Latin letters $A$, $B$ run from $-1$ to $D+1$.
Introduce $J_{AB}$ as follows:
\begin{eqnarray}\label{defofJ}
J_{AB}=\left\{\begin{array}{ll}  J_{\mu\nu} & \hbox{if $A=\mu$,
$B=\nu$}\cr  A_\mu & \hbox{if $A=\mu$, $B=D+1$}\cr M_\mu& \hbox{if
$A=\mu$, $B=-1$}\cr \Gamma_\mu & \hbox{if $A=\mu$, $B=0$}\cr T &
\hbox{if $A=D+1$, $B=-1$}\cr \Gamma_{D+1} & \hbox{if $A=D+1$,
$B=0$}\cr \Gamma_{-1} & \hbox{if $A=-1$, $B=0$}\cr -J_{BA} &
\hbox{if $A>B$}\cr 0 & \hbox{if $A=B$}.\cr
\end{array}\right.
\end{eqnarray}

\begin{Thm}\label{keyr} Assume $\mu=0$ or $1/2$.
Let $C^\infty({\mathcal S}^{2\mu})$ be the space of smooth sections
of ${\mathcal S}^{2\mu}$, and $J_{AB}$ be defined by Eq.
(\ref{defofJ}).

1) As operators on $C^\infty({\mathcal S}^{2\mu})$, $J_{AB}$'s
satisfy the following commutation relations:
\begin{eqnarray}\label{cmtr} [J_{AB},
J_{A'B'}]=-i\eta_{AA'}J_{BB'}-i\eta_{BB'}J_{AA'}+i\eta_{AB'}J_{BA'}+i\eta_{BA'}J_{AB'}
\end{eqnarray}
where the indefinite metric tensor $\eta$ is ${\mr
{diag}}\{++-\cdots-\}$ relative to the following order: $-1$, $0$,
$1$, \ldots, $2n+1$ for the indices.

2) As operators on $C^\infty({\mathcal S}^{2\mu})$,
\begin{eqnarray} \{J_{AB}, {J^A}_C\}:=J_{AB}{J^A}_C+{J^A}_CJ_{AB}=-2a\eta_{BC}
\end{eqnarray}where $a=n-1/2-c$.
\end{Thm}
\begin{proof} The proof of this theorem can be transplanted directly from the
one for theorem 2 in Ref. \cite{MZ07}. In fact, one just needs to
modify the proof there by doing the followings: 1) replace $n$ in
Ref. \cite{MZ07} by $n-1/2$ which is always $D-1\over 2$, 2) quote
part 2) of lemma \ref{lemma} as follows: \emph{Identity
\begin{eqnarray}\label{anotheralgReason}
r^2F_{\lambda\alpha}F_{\lambda\beta}=c\left({1\over
r^2}\delta_{\alpha\beta}-{x_\alpha x_\beta\over
r^4}\right)+i{D-3\over 2}F_{\alpha\beta}
\end{eqnarray} holds in representation ${\bf s}^{2\mu}$ of
$\frk{so}_0(D-1)$}. \end{proof}

Note that, when $D$ is even, identity (\ref{anotheralgReason}) is
not true unless $\mu=0$ or $1/2$; this is another evidence for the
nonexistence of the even-dimensional generalized MICZ-Kepler problem
with magnetic charge higher than $1/2$.

\section{Representation theoretical aspects --- the preliminary part}\label{Rep}
The main objective in the rest of this paper is to show that the
algebraic direct sum $\mathcal H$ of the energy eigenspaces of a
generalized MICZ-Kepler problem in dimension $2n$ is a unitary
highest weight $(\frk{g}, K)$-module where $\frk{g}=\frk{so}(2n+3)$
and $K=\mr{Spin}(2)\times \mr{Spin}(2n+1)$. Along the way, we prove
Theorem \ref{main}.

We can label the generators of $\frk{g}_0$ (the Lie algebra of
$\widetilde{\mr{Spin}}(2,2n+1)$) as follows:
$$
M_{AB}=-M_{BA} \quad \mbox{for $A, B = -1, 0, 1,  \ldots, 2n+1$ }
$$ where in the $(2n+3)$-dimensional defining representation, the
matrix elements of $M_{AB}$ are given by
$$
[M_{AB}]_{JK}=-i(\eta_{AJ}\eta_{BK}-\eta_{BJ}\eta_{AK})
$$ with the indefinite metric tensor $\eta$ being ${\mr
{diag}}\{++-\cdots-\}$ relative to the following order: $-1$, $0$,
$1$, \ldots, $2n+1$ for the indices.

One can easily show that
\begin{eqnarray}\label{cmtrm} \hskip 20pt[M_{AB},
M_{A'B'}] =
i(\eta_{AA'}M_{BB'}+\eta_{BB'}M_{AA'}-\eta_{AB'}M_{BA'}-\eta_{BA'}M_{AB'}).
\end{eqnarray}

In view of the sign difference between the right hand sides of Eqs.
(\ref{cmtr}) and (\ref{cmtrm}), we define the representation
$(\tilde\pi,C^\infty({\mathcal S}^{2\mu}))$ of $\frk{g}$ as follows:
for $\psi\in C^\infty({\mathcal S}^{2\mu})$, \begin{eqnarray}\fbox{$
\tilde\pi(M_{AB})(\psi)= -\hat J_{A B}\psi$}
\end{eqnarray}
where, by definition, $\hat J_{AB}:={1\over {\sqrt r}}J_{AB}\sqrt
r$. Note that $\tilde \pi$ defines a representation for $\frk{g}$
only when $\mu=0$ or $1/2$.

However, what is really relevant for us is just a subspace of
$C^\infty({\mathcal S}^{2\mu})$, i.e., $\mathcal H$. Actually, the
story is bit more involved: what is really invariant under
$\tilde\pi$ is not $\mathcal H$, but a twisted version of $\mathcal
H$ which is denoted by $\tilde {\mathcal H}$ later; and there is a
twist linear equivalence
$$\tau:\;\mathcal H \to \tilde{\mathcal H}$$ which preserves the $L^2$-norm,
such that, viewing $\tau$ as an equivalence of representations, we
get representation $(\pi, \mathcal H)$. Because of this intricacy,
we shall devote the next two subsections to some preparations.

\subsection{Review of the (bound) energy eigenspaces} The bound eigen-states (i.e., \emph{$L^2$ eigen-sections} of the Hamiltonian) of
the generalized MICZ-Kepler problems have been analyzed in section
5.2 of Ref. \cite{meng05} by using the classical analytic method
with the help of the representation theory for compact Lie groups.
Recall that the (bound) energy spectrum is
\begin{eqnarray}
E_I=-{1\over 2(I+n+\mu-1/2)^2}
\end{eqnarray} where $I=0, 1, 2, \cdots$.

Denote by ${\mathcal S}^{2\mu}|_{\mr{S}^{2n-1}}$ the restriction
bundle of ${\mathcal S}^{2\mu}$ to the unit sphere $\mr{S}^{2n-1}$.
By the Frobenius reciprocity plus a branching rule for
$(\mr{Spin}(2n), \mr{Spin}(2n-1))$, one has
\begin{eqnarray}\label{brch0}
L^2({\mathcal S}^{2\mu}|_{\mr{S}^{2n-1}} )=\left\{
\begin{array}{ll}
\hat\bigoplus_{l\ge 0} ({\ms R}_{l}^+\oplus{\ms R}_{l}^-) & \mbox{if $\mu=1/2$}\\
\\
\hat\bigoplus_{l\ge 0} {\ms R}_{l}^0 & \mbox{if $\mu=0$}
\end{array}\right.
\end{eqnarray} where ${\ms R}_{l}^\pm$ is the irreducible representation of
$\mr{Spin}(2n)$ with highest weight $$(l+1/2, 1/2, \cdots,
1/2,\pm 1/2)$$ and ${\ms R}_{l}^0$ is the irreducible representation
of $\mr{Spin}(2n)$ with highest weight $(l, 0, \cdots, 0)$.
One can also show that, the infinitesimal action of $\mr{Spin}(2n)$
on $C^\infty({\mathcal S}^{2\mu})$ is just the restriction of
$\tilde\pi$ to $\mr{span}_{\bb R}\{M_{\alpha\beta}\mid
1\le\alpha<\beta\le 2n\}=\frk{so}_0(2n)$. It is then clear that
$\tilde\pi(M_{\alpha\beta})$'s act only on the angular part of the
wave sections --- a consequence which can also be deduced from the
fact that $\hat J_{\alpha\beta}$'s commute with the multiplication
by a smooth function of $r$.

Let
$$\fbox{$l_\mu=l+\mu+n-3/2$}\,.$$ For $\sigma\in\{+, -, 0\}$, we let $\{Y_{l\bf
m}^\sigma(\Omega)\}_{{\bf m}\in {\mathcal I}^\sigma(l)}$ be an
orthornormal (say Gelfand-Zeltin) basis for ${\ms R}_{l}^\sigma$.
Then, an orthornormal basis for the energy eigenspace $\ms H_I$ with
energy $E_I$ is
\begin{eqnarray}
\{\psi_{kl \bf m}^\sigma:=R_{kl_\mu}(r)Y_{l\bf m}^\sigma(\Omega)\mid
k+l=I+1, k\ge 1, {\bf m}\in {\mathcal I}^\sigma(l), l\ge 0,
\sigma\in \{+,-\}\}\nonumber
\end{eqnarray} if $\mu=1/2$, and is
\begin{eqnarray}
\{\psi_{kl \bf m}^0:=R_{kl_\mu}(r)Y_{l\bf m}^0(\Omega)\mid k+l=I+1,
{\bf m}\in {\mathcal I}^0(l), k\ge 1, l\ge 0\}\nonumber
\end{eqnarray} if $\mu=0$.
Here $R_{kl_\mu}\in L^2({\bb R}_+, r^{2n-1}\,dr)$ is a square
integrable (with respect to measure $r^{2n-1}\,dr$) solution of the
radial Schr\"{o}dinger equation:
\begin{eqnarray}\label{rSchEq}
\left(-{1\over 2r^{2n-1}}\partial_r
r^{2n-1}\partial_r+{l_\mu(l_\mu+1)-(n-{1\over 2})(n-{3\over 2})\over
2r^2}-{1\over r}\right)R_{kl_\mu}=E_{k-1+l}R_{kl_\mu}.\cr
\end{eqnarray}
Note that $R_{kl_\mu}$ is of the form
$$r^{-n+1/2}y_{kl_\mu}(r)\exp\left(-{r\over k+l_\mu}\right)$$
with $y_{kl_\mu}(r)$ satisfying Eq.
\begin{eqnarray}
\left( {d^2\over dr^2} -{2\over k+l_\mu} {d\over dr}+\left[{2\over
r}-{l_\mu(l_\mu+1)\over
r^2}\right]\right)y_{kl_\mu}(r)=0.\end{eqnarray}In term of the
generalized Laguerre polynomials,
$$
y_{kl_\mu}(r)=c(k,l) r^{l_\mu+1}L^{2l_\mu+1}_{k-1}\left({2\over
k+l_\mu}r\right)
$$
where $c(k,l)$ is a constant, which can be uniquely determined by
requiring $c(k,l)>0$ and $\int_0^\infty |R_{kl_\mu}(r)|^2
r^{2n}dr=1$.

We are now ready to state the following remark.
\begin{rmk}\label{rmk1}
1) $\ms H_I$ is the space of square integrable solutions of Eq. $
H\psi = E_I\psi$.

2) As representation of $\frk{so}(2n)$,
\begin{eqnarray}
{\ms H}_I=\left\{\begin{array}{ll} \bigoplus_{l=0}^I
\left(D_l^+\oplus D_l^-\right) & \mbox{if $\mu=1/2$}\\
\\
\bigoplus_{l=0}^I  D_l^0 & \mbox{if $\mu=0$}
\end{array}\right.
\end{eqnarray} where $D_{l}^\pm=\mr{span}\{\psi_{(I-l+1)l\bf m}^\pm \mid
{\bf m}\in {\mathcal I}^\pm(l)\}$ is the highest weight module with
highest weight $(l+1/2, 1/2, \cdots,1/2, \pm 1/2)$, and
$D_{l}^0=\mr{span}\{\psi_{(I-l+1)l\bf m}^0 \mid {\bf m}\in {\mathcal
I}^0(l)\}$ is the highest weight module with highest weight $(l, 0,
\cdots, 0)$.

3) $\{{\ms H}_I \mid I= 0, 1, 2, \ldots\}$ is the complete set of
(bound) energy eigenspaces.

\end{rmk}

For the completeness of this review, we state part of Theorem 2 from
Ref. \cite{meng05} below:
\begin{Th} Let $n>1$ be an integer and $\mu=0$ or $1/2$. For the $2n$-dimensional generalized
MICZ-Kepler problem with magnetic charge $\mu$, the following
statements are true:

1) The negative energy spectrum is
$$
E_I=-{1/2\over (I+n+\mu-1/2)^2}
$$ where $I=0$, $1$, $2$, \ldots;

2) The Hilbert space $\ms H(\mu)$ of negative-energy states admits a
linear $\mr{Spin}(2n+1)$-action under which there is a decomposition
$$
{\ms H}(\mu)=\hat\bigoplus _{I=0}^\infty\,{\ms H}_I
$$ where ${\ms H}_I$ is the irreducible $\mr{Spin}(2n+1)$-module
with highest weight $(I+\mu,\mu, \cdots, \mu)$;

3) The linear action in part 2) extends the manifest linear action
of $\mr{Spin}(2n)$, and ${\ms H}_I$ in part 2) is the energy
eigenspace with eigenvalue $E_I$ in part 1).
\end{Th}

It was shown in Ref. \cite{meng05} that the bound eigen-states are
precisely the ones with negative energy eigenvalues.

\subsection{Twisting} As we said before, because of the technical
intricacy, we need to introduce the notion of twisting. Let us start
with the listing of some important spaces used later:
\begin{itemize}
\item $\ms H_I$ --- the $I$-th bound energy eigenspace;
\item $\mathcal H$ --- the algebraic direct sum of all bound energy eigenspaces;
\item $\ms H$ or $\ms H(\mu)$ --- the completion of $\mathcal H$ under the standard $L^2$-norm;
\item ${\mathcal H}_{l\bf m}^\sigma$ --- the subspace of $\mathcal H$ spanned by $\{\psi_{kl\bf m}^\sigma\,|\, k\ge
1,\; \mbox{$l$, $\bf m$, $\sigma$ fixed}\}$;
\item ${\ms H}_{l\bf m}^\sigma$ --- the completion of ${\mathcal H}_{l\bf
m}^\sigma$ under the standard $L^2$-norm.
\end{itemize}
Note that these spaces are all endowed with the unique hermitian
inner product which yields the standard $L^2$-norm, i.e.,
\begin{eqnarray}
\fbox{$\langle\psi,\phi\rangle:=\displaystyle\int_{{\bb
R}^{D}_*}(\psi, \phi)\, d^{D}x$}
\end{eqnarray} where $(\psi, \phi)$ is the point-wise hermitian inner product
and $d^Dx$ is the Lebesgue measure.

It is clear from the previous section that
\begin{eqnarray}
{\ms B}:=\left\{
\begin{array}{ll} \{\psi_{kl\bf m}^\sigma \mid k\ge 1, {\bf m}\in {\mathcal I}^\sigma(l),
l\ge 0, \sigma\in\{+,-\}\} &
\mbox{if $\mu=1/2$}\\
\\
\{\psi_{kl\bf m}^0\mid k\ge 1, {\bf m}\in {\mathcal I}^0(l), l\ge
0\} & \mbox{if $\mu=0$}
\end{array}\right.
\end{eqnarray} is an orthonormal basis for both ${\mathcal H}$ and $\ms H$.

To study the action of $\hat J_{AB}$'s, we need to ``twist" $\ms B$,
$\ms H_I$, ${\mathcal H}_{l\bf m}^\sigma$, ${\ms H}_{l\bf
m}^\sigma$, $\mathcal H$ and $\ms H$ to get $\tilde{\ms B}$, $\tilde
{\ms H}_I$, $\tilde{\mathcal H}_{l\bf m}^\sigma$, $\tilde{\ms
H}_{l\bf m}^\sigma$, $\tilde {\mathcal H}$ and $\tilde {\ms H}$
respectively. It suffices to twist the elements of ${\ms B}$. Let
$\tau$: $\ms B\to\tilde{\ms B}$ be defined as follows:
\begin{eqnarray}\label{psitilde}
\tau(\psi_{kl\bf m}^\sigma)(r,
\Omega)&:=&(k+l_\mu)\,e^{-i\theta_{k+l_\mu}\hat T}\left({1\over
\sqrt r}\psi_{kl\bf m}^\sigma(r, \Omega)\right)\cr  &=&
(k+l_\mu)^{n+1/2}{1\over \sqrt r}\psi_{kl\bf m}^\sigma((k+l_\mu)r,
\Omega)\cr & \varpropto & r^{l +\mu-{1\over
2}}\,L^{2l_\mu+1}_{k-1}(2r)\,e^{-r}\,Y_{l\bf m}^\sigma(\Omega)
\end{eqnarray}
where $\hat T={1\over \sqrt r} T \sqrt r$, and $\theta_I=-\ln I$ for
any positive number $I$. For simplicity, we write $\tau(\psi_{kl\bf
m}^\sigma)$ as $\tilde\psi_{kl\bf m}^\sigma$. One can check that
$$
\int_{{\bb R}^{D}_*}(\tilde \psi_{kl\bf m}^\sigma, \tilde
\psi_{kl\bf m}^\sigma)\, d^{D}x=\int_{{\bb R}^{D}_*}(\psi_{kl\bf
m}^\sigma, \psi_{kl\bf m}^\sigma)\, d^{D}x=1.
$$

By using Eq. (\ref{psitilde}) and the orthogonality identities for
the generalized Laguerre polynomials, one can see that
$\tilde\psi_{kl\bf m}^\sigma$ is orthogonal to $\tilde\psi_{k'l\bf
m}^\sigma$ when $k\neq k'$.

It is now clear how to twist all the relevant spaces listed in the
beginning of this subsection. For example,
\begin{eqnarray}
{\tilde {\ms H}}_I:=\left\{\exp(-i\theta_{I_\mu+1}\hat
T)\left({1\over \sqrt r}\psi\right)\,|\, \psi \in \ms H_I\right\}.
\end{eqnarray}
Since $\ms H_I$ is spanned by $\{\psi_{kl\bf m}^\sigma\mid k+l=I+1,
k\ge 1, {\bf m}\in {\mathcal I}^\sigma(l), l\ge 0,
\sigma\in\{+,-\}\}$ if $\mu=1/2$, or by $ \{\psi_{kl\bf m}^0\mid
k+l=I+1, k\ge 1, {\bf m}\in {\mathcal I}^0(l), l\ge 0\}$ if $\mu=0$,
it follows that $\tilde{\ms H}_I$ is spanned by
$$\{\tilde \psi_{kl\bf m}^\sigma\mid k+l=I+1, k\ge 1, {\bf m}\in
{\mathcal I}^\sigma(l), l\ge 0, \sigma\in\{+,-\}\} \mbox{ if
$\mu=1/2$},$$ or by
$$\{\tilde \psi_{kl\bf m}^0\mid k+l=I+1, k\ge 1, {\bf m}\in
{\mathcal I}^0(l), l\ge 0\} \mbox{ if $\mu=0$}.$$

We shall call $\tilde{\ms H}(\mu)$ the twisted Hilbert space of the
bound states for the $2n$-dimensional generalized MICZ-Kepler
problem with magnetic charge $\mu$. Remark that the twisting
map\footnote{It has a \underline{basis-free} description: for
$\psi\in \ms H_I$, $\tau(\psi)(r, \Omega)=(I_\mu+1)^{n+1/2}{1\over
\sqrt r}\psi((I_\mu+1)r, \Omega)$. }
\begin{eqnarray}\tau:\; {\ms
H}(\mu)\to \tilde{\ms H}(\mu)
\end{eqnarray}
is the unique linear isometry which sends $\psi_{kl\bf m}^\sigma$ to
$\tilde \psi_{kl\bf m}^\sigma$; moreover, $\tau$ maps all relevant
subspaces of ${\ms H}(\mu)$ isomorphically onto the corresponding
relevant twisted subspaces. Note that $\hat J_{\alpha\beta}={1\over
\sqrt r} J_{\alpha\beta} \sqrt r= J_{\alpha\beta}$ obviously acts on
$\tilde {\ms H}_I$ as hermitian operator, so
$\frk{r}:=\mr{span}\{M_{\alpha\beta}\, |\, 1\le \alpha< \beta\le
2n\}=\frk{so}(2n)$  acts unitarily on $\tilde {\ms H}_I$ via $\tilde
\pi$.

Recall that for non negative integer $I$, we use $I_\mu$ to denote
$I+n+\mu-3/2$.
\begin{Prop}\label{Prmk2}
1) $\tilde\psi_{kl\bf m}^\sigma$ is an eigenvector of $\hat
\Gamma_{-1}$ with eigenvalue $k+l_\mu$.

2) $\tilde{\ms H_I}$ is the space of square integrable solutions of
Eq. $\hat\Gamma_{-1}\psi = (I_\mu+1)\psi$.

3) $\hat\Gamma_{-1}$ is a self-adjoint operator on $\tilde{\ms
H}(\mu)$ and $\tilde{\ms H_I}$ is the eigenspace of
$\hat\Gamma_{-1}$ with eigenvalue $I_\mu+1$.

4) As representation of $\frk{r}$,
\begin{eqnarray}
\tilde {\ms H}_I=\left\{ \begin{array}{ll}\bigoplus_{l=0}^I
\left(\tilde D_l^+\oplus \tilde D_l^-\right) & \mbox{if $\mu=1/2$}\\
\\
\bigoplus_{l=0}^I \tilde D_l^0 & \mbox{if $\mu=0$}
\end{array}\right.
\end{eqnarray} where $\tilde D_l^\pm=\mr{span}\{\tilde \psi_{(I-l+1)l\bf m}^\pm\mid
{\bf m}\in {\mathcal I}^\pm(l)\}$ is the highest weight module with
highest weight $(l+1/2, 1/2, \cdots, 1/2,\pm 1/2)$, and $\tilde
D_l^0=\mr{span}\{\tilde \psi_{(I-l+1)l\bf m}^0\mid {\bf m}\in
{\mathcal I}^0(l)\}$ is the highest weight module with highest
weight $(l, 0, \cdots, 0)$.

5) $\tilde{\ms H}(\mu)= L^2({\mathcal S}^{2\mu})$.

\end{Prop}
\begin{proof}
1) The proof is based on the ideas from Ref. \cite{Barut71}. Since
\begin{eqnarray}\label{eigeneq1}
H\psi_{kl\bf m}^\sigma=E_{k+l-1}\psi_{kl\bf m}^\sigma,
\end{eqnarray}
we have $ \sqrt r(H-E_{k+l-1})\psi_{kl\bf m}^\sigma=0$ which can be
rewritten as
$$
({1\over 2}\hat X-1-E_{k+l-1}\hat Y)({1\over \sqrt r}\psi_{kl\bf
m}^\sigma)=0
$$ where $X$ and $Y$ are given by Eq. (\ref{def1}). In terms
of $\hat \Gamma_{-1}$ and $\hat \Gamma_{D+1}$, we can recast the
above equation as
$$
\left(({1\over 2}-E_{k+l-1})\hat \Gamma_{-1}+({1\over
2}+E_{k+l-1})\hat \Gamma_{D+1}-1\right)({1\over \sqrt r}\psi_{kl\bf
m}^\sigma)=0
$$
Plugging $\psi_{kl\bf m}^\sigma={1\over k+l_\mu}\sqrt r
e^{i\theta_{k+l_\mu}\hat T}\left(\tilde \psi_{kl\bf
m}^\sigma\right)$ into the above equation and using identities
\begin{eqnarray}\left\{
\begin{array}{rcl} e^{-i\theta \hat T}\,\hat\Gamma_{-1}\,e^{i\theta \hat T} & = &
\cosh\theta\,\hat\Gamma_{-1}+\sinh\theta\,\hat\Gamma_{D+1}\\
\\
e^{-i\theta \hat T}\,\hat\Gamma_{D+1}\,e^{i\theta \hat T} & = &
\sinh\theta\,\hat\Gamma_{-1}+
\cosh\theta\,\hat\Gamma_{D+1},\end{array}\right.
\end{eqnarray}
we arrive at the following equation:
\begin{eqnarray}\label{eigeneq2}
\fbox{$\hat \Gamma_{-1}\tilde\psi_{kl\bf
m}^\sigma=(k+l_\mu)\tilde\psi_{kl\bf m}^\sigma$}\,.
\end{eqnarray}

2) Note that the Barut-Bornzin process going from Eq.
(\ref{eigeneq1}) to Eq. (\ref{eigeneq2}) is completely reversible.
Therefore, part 2) is just a consequence of part 1) of Remark
\ref{rmk1}.

3) Note that $\hat \Gamma_{-1}$ is defined on the dense linear
subspace $\tilde{\mathcal H}$ of $\tilde{\ms H} (\mu)$. It is easy
to check that $$\langle\tilde\psi_{k'l'\bf m'}^\sigma, \hat
\Gamma_{-1}\tilde\psi_{kl\bf m}^\sigma\rangle=\langle\hat
\Gamma_{-1}\tilde\psi_{k'l'\bf m'}^\sigma, \tilde\psi_{kl\bf
m}^\sigma\rangle$$ for any $\tilde\psi_{kl\bf m}^\sigma$ and
$\tilde\psi_{k'l'\bf m'}^\sigma$. Therefore, $\hat \Gamma_{-1}$ (To
be precise, it should be its closure) is a self-adjoint operator on
$\tilde{\ms H} (\mu)$. In view of part 2), $\tilde{\ms H}_I$ is the
eigenspace of $\hat\Gamma_{-1}$ with eigenvalue $I_\mu+1$.

4) This part is clear due to part 2) of Remark \ref{rmk1}.

5) Recall that $\tilde\psi_{kl\bf m}^\sigma(r,\Omega)=\tilde
R_{kl_\mu}(r)\,Y_{l\bf m}^\sigma(\Omega)$ where $\tilde
R_{kl_\mu}(r)\propto r^{l +\mu-{1\over
2}}\,L^{2l_\mu+1}_{k-1}(2r)\,e^{-r}$.  By the well-known property
for the generalized Laguerre polynomials, for any $l\ge 0$,
$\{\tilde R_{kl_\mu}\}_{k=1}^\infty$ form an orthonormal basis for
$L^2(\bb R_+, r^{2n-1}\,dr)$.

By virtue of Theorem II. 10 of Ref. \cite{Reed&Simon} and Eq.
(\ref{brch0}),
\begin{eqnarray}
L^2({\mathcal S}^{2\mu}) & = & L^2(\bb R_+, r^{2n-1}\,dr)\otimes
L^2({\mathcal S}^{2\mu}\mid_{\mr{S}^{2n-1}})\cr & = &\left\{
\begin{array}{ll} \hat \bigoplus_{l=0}^\infty L^2(\bb R_+,
r^{2n-1}\,dr)\otimes \left({\ms R}_l^+\oplus {\ms R}_l^-\right) &
\mbox{if
$\mu=1/2$}\\
\\
\hat \bigoplus_{l=0}^\infty  L^2(\bb R_+, r^{2n-1}\,dr)\otimes {\ms
R}_l^0 & \mbox{if $\mu=0$.}\end{array}\right.\nonumber
\end{eqnarray} Therefore, $\tilde {\ms B}$ is an orthonormal basis
for $L^2({\mathcal S}^{2\mu})$, consequently $\tilde{\ms H}(\mu)=
L^2({\mathcal S}^{2\mu})$.

\end{proof}
We end this subsection with
\begin{rmk}\label{rmk2}
$\tilde{\ms H}_I$ is the eigenspace of $\tilde\pi(H_0)$ with
eigenvalue $-(I_\mu+1)$. Here $\tilde\pi(H_0)=-\hat \Gamma_{-1}$ is
viewed as an endomorphism of $\tilde {\mathcal H}$.
\end{rmk}

\section{Representation theoretical aspects --- the final part}
For the remainder of this paper, we assume the magnetic charge $\mu$
is either $0$ or $1/2$. We start with some notations:
\begin{itemize}
\item $G=\widetilde{\mr{Spin}}(2, 2n+1)$ --- the 4-fold
cover of $\mr{SO}_0(2, 2n+1)$ characterized by the homomorphism
$\pi_1(\mr{SO}_0(2, 2n+1))=\bb Z\oplus \bb Z_2\to \bb Z_2\oplus \bb
Z_2$ sending $(a, b)$ to $(\bar a, b)$;
\item $\frk{g}_0$ --- the Lie algebra of $\widetilde{\mr{Spin}}(2, 2n+1)$;
\item $\frk{g}$ --- the complexfication of $\frk{g}_0$, so $\frk{g}=\frk{so}(2n+3)$;
\item $H_0$ --- defined to be $M_{-1, 0}$;
\item $H_j$ --- defined to be $-M_{2j-1, 2j}$ for $1\le j\le n$;
\item $K:=\mr{Spin}(2)\times\mr{Spin}(2n+1)$ --- a maximal compact subgroup of $\widetilde{\mr{Spin}}(2, 2n+1)$;
\item $\frk{k}_0$ --- the Lie algebra of $K$;
\item $\frk{k}$ --- the complexfication of $\frk{k}_0$, so $\frk{k}=\frk{so}(2)\oplus\frk{so}(2n+1)$;
\item $\frk{r}$ --- the subalgebra of $\frk{g}$ generated by
$\{M_{AB}\,|\, 1\le A < B \le 2n\}$, so $\frk{r}_0:=\frk{g}_0\cap
\frk{r}=\frk{so}_0(2n)$;
\item $\frk{s}$ --- the subalgebra of $\frk{g}$ generated by  $\{M_{AB}\,|\, 1\le A < B
\le 2n+1\}$, so $\frk{s}_0:=\frk{g}_0\cap \frk{s}=\frk{so}_0(2n+1)$;
\item $\frk{sl}(2)$ --- the subalgebra of $\frk{g}$ generated by  $M_{-1, D+1}$, $M_{0, D+1}$ and $M_{-1,0}$, so $\frk{sl}_0(2):=\frk{g}_0\cap
\frk{sl}(2)=\frk{so}_0(2,1)$;
\item $U(\frk{sl}(2))$ --- the universal enveloping algebra of
$\frk{sl}(2)$.
\end{itemize}

\subsection{$\tilde{\mathcal H}$ is a unitary highest weight
Harish-Chandra module} The goal of this subsection is to show that
$(\tilde\pi, \tilde{\mathcal H})$ is a unitary highest weight
$(\frk{g}, K)$-module.

\begin{Prop}\label{prop1}
1) Each $\tilde\pi(M_{AB})$ maps $\tilde{\mathcal H}$ into
$\tilde{\mathcal H}$, so $(\tilde \pi, \tilde{\mathcal H})$ is a
representation of $\frk{g}$.

2) Each $\tilde\pi(M_{AB})$ is a hermitian operator on
$\tilde{\mathcal H}$, so $(\tilde \pi, \tilde{\mathcal H})$ is a
unitary representation of $\frk{g}$.

3) $(\tilde \pi|_{\frk{sl}(2)}, \tilde{\mathcal H}_{l\bf m}^\sigma)$
is the discrete series representation of $\frk{sl}(2)=\frk{so}(2,1)$ with
highest weight $-l_\mu-1$.

\end{Prop}
\begin{proof} 1) We follow the convention of Ref. \cite{georgi82} for describing the root space of $\frk{g}=\frk{so}(2n+3)$.
Take as a basis of the Cartan sub-algebra of $\frk{g}$ the following
elements:
$$
H_0=M_{-1, 0},\quad H_j=-M_{2j-1, 2j},\quad \mbox{$j=1, \cdots, n$}.
$$
Let $\eta, \eta'=\pm 1$. We take the following root vectors:
\begin{eqnarray}
E_{\eta e^j+\eta' e^k} &= &{1\over 2}\left(M_{2j-1, 2k-1}+i\eta
M_{2j, 2k-1}+i\eta'M_{2j-1, 2k}-\eta\eta' M_{2j, 2k}\right)\cr
& &\quad\quad\quad
\mbox{for $0\le j<k\le n$},\cr E_{\eta e^j} & = & {1\over \sqrt
2}\left(M_{2j-1, 2n+1}+i\eta M_{2j, 2n+1}\right)\quad \mbox{for
$0\le j\le n$}.\nonumber
\end{eqnarray} This way we obtain a Cartan basis for $\frk{g}$.
Therefore, for $\psi_I\in \tilde{\ms H}_I$, we have
\begin{eqnarray}\label{eigeneq}
\tilde \pi(H_0)(\tilde\pi(E_\alpha)(\psi_I))
 & = & (-I_\mu-1+\alpha_0)\tilde\pi(E_\alpha)(\psi_I)\cr & = &
 (-(I-\alpha_0)_\mu-1)\tilde\pi(E_\alpha)(\psi_I)
\end{eqnarray}
where $\alpha_0$ (which can be $0$, or $-1$ or  $1$) is the $0$-th
component of $\alpha$. It is not hard to see that
$\tilde\pi(E_\alpha)(\psi_I)$ is square integrable \footnote{The
convergence of the integral near infinity is clear because of the
exponential decay as $r\to \infty$. The convergence of the integral
near the origin of $\bb R^D_*$ is clear from the counting of powers
of $r$.}, so in view of part 2) of Proposition \ref{Prmk2}, Eq.
(\ref{eigeneq}) implies that $\tilde\pi(E_\alpha)(\psi_I)\in
{\tilde{\ms H}}_{I-\alpha_0}$. (Here ${\ms H}_{-1}=0$.) Therefore,
$\tilde\pi(E_\alpha)$ maps any $\tilde{\ms H}_I$, hence
$\tilde{\mathcal H}$, into $\tilde{\mathcal H}$. By a similar
argument, one can show that $\tilde\pi(H_i)$ maps $\tilde{\mathcal
H}$ into itself. Since $H$'s and $E$'s form a basis for $\frk{g}$,
this implies that $\tilde\pi(M_{AB})$ maps $\tilde{\mathcal H}$ into
itself.

2) It is equivalent to checking that each $\hat J_{AB}:={1\over
{\sqrt r}}J_{AB}\sqrt r$ is an hermitian operator on
$\tilde{\mathcal H}$. First of all, it is not hard to see that, when
${\mathcal O} =\pi_\alpha, r, {1\over r}, \sqrt r, {1\over \sqrt
r}$, we always have
\begin{eqnarray}\label{hermitian}\langle \psi_1, {\mathcal
O}\psi_2\rangle=\langle {\mathcal O}\psi_1,
\psi_2\rangle\end{eqnarray} for any $\psi_1$, $\psi_2$ in
$\tilde{\mathcal H}$. It is equally easy to see that Eq.
(\ref{hermitian}) is always true for any $\psi_1$, $\psi_2$ in
$\tilde{\mathcal H}$ when ${\mathcal O}$ is $\hat\Gamma_\alpha
=\sqrt r \pi_\alpha \sqrt r$, $\hat X =\sqrt r \pi^2 \sqrt r+{c\over
r}$, or $\hat Y=r$. It is then clear from definitions (\ref{def1})
and (\ref{def2}) that Eq. (\ref{hermitian}) is always true for any
$\psi_1$, $\psi_2$ in $\tilde{\mathcal H}$ when ${\mathcal O} =\hat
J_{AB}$.

3) Let us first show that $\tilde\pi(M_{-1, D+1})$, $\tilde\pi(M_{0,
D+1})$ and $\tilde\pi(M_{-1,0})$ map each $\tilde\psi_{kl\bf
m}^\sigma$ into $\tilde{\mathcal H}_{l\bf m}^\sigma$, so they indeed
map $\tilde{\mathcal H}_{l\bf m}^\sigma$ into $\tilde{\mathcal
H}_{l\bf m}^\sigma$. This is obvious for $\tilde\pi(M_{-1,0})$
because $\tilde\pi(M_{-1,0})(\tilde\psi_{kl\bf m}^\sigma)= -\hat
\Gamma_{-1}\tilde\psi_{kl\bf m}^\sigma=-(k+l_\mu)\tilde\psi_{kl\bf
m}^\sigma$. Next, we introduce
$$E_\pm={1\over \sqrt 2}(M_{-1, D+1}\pm iM_{0, D+1}),$$then one can
check from Eq. (\ref{cmtrm}) that $[M_{-1,0}, E_\pm]=\pm E_\pm$.
Therefore
$$
\tilde\pi(M_{-1,0})(\tilde\pi(E_\pm)( \tilde\psi_{kl\bf
m}^\sigma))=(-k-l_\mu\pm 1)\tilde\pi( E_\pm)( \tilde\psi_{kl\bf
m}^\sigma),
$$ where $\tilde\pi( E_\pm)={1\over \sqrt 2}(\hat T\pm i\hat\Gamma_{D+1})$.
It is not hard to see that $ \tilde\pi( E_\pm)( \tilde\psi_{kl\bf
m}^\sigma)$ is square integrable. In view of part 2) of Proposition
\ref{Prmk2}, we conclude that $\tilde \pi( E_\pm)( \tilde\psi_{kl\bf
m}^\sigma)$ must be proportional to $\tilde\psi_{(k\mp 1)l\bf
m}^\sigma$. (Here, by convention, $\tilde\psi_{0l\bf
m}^\sigma=0$.) Therefore, operators $\tilde \pi( E_\pm)$ map
$\tilde\psi_{kl\bf m}^\sigma$ into $\tilde{\mathcal H}_{l\bf
m}^\sigma$. This proves that $(\tilde \pi|_{\frk{sl}(2)},
\tilde{\mathcal H}_{l\bf m}^\sigma)$ is a representation of
$\frk{sl}(2)$.

In view of the fact that $\tilde\pi(M_{-1, 0})(\tilde \psi_{1l\bf
m}^\sigma)=-(l_\mu+1)\tilde \psi_{1l\bf m}^\sigma\neq 0$, we
conclude that $U(\frk{sl}(2))\cdot \tilde \psi_{1l\bf m}^\sigma$ is
a \emph{nontrivial} unitary highest weight representation of the
non-compact real Lie algebra $\frk{sl}_0(2)$, hence must be the
discrete series representation with highest weight $-(l_\mu+1)$.
Since $U(\frk{sl}(2))\cdot \tilde \psi_{1l\bf m}^\sigma\subset
\tilde {\mathcal H}_{l\bf m}^\sigma$, and $$\dim (\tilde{\ms
H}_I\,\cap\,U(\frk{sl}(2))\cdot \tilde \psi_{1l\bf m}^\sigma)=\dim
(\tilde{\ms H}_I\,\cap\, \tilde {\mathcal H}_{l\bf m}^\sigma)$$ for
all $I\ge 0$, we conclude that $U(\frk{sl}(2))\cdot \tilde
\psi_{1l\bf m}^\sigma= \tilde {\mathcal H}_{l\bf m}^\sigma$.
Therefore, $\tilde {\mathcal H}_{l\bf m}^\sigma$ is a unitary
highest weight $\frk{sl}(2)$-module with highest weight $-l_\mu-1$,
which in fact is a unitary highest weight $(\frk{sl}(2),
\mr{Spin}(2))$-module. Then $\tilde {\ms H}_{l\bf m}^\sigma$ must be
the discrete series representation of $\mr {Spin}(2,1)$ with highest
weight $-l_\mu-1$.

\end{proof}

To continue the discussion on representations, we prove the
following proposition.
\begin{Prop}\label{prop2}

1) $(\tilde\pi|_{\frk{s}}, \tilde{\ms H}_I)$ is an irreducible
unitary representation of $\frk{s}$, in fact, it is the highest
weight representation with highest weight $(I+\mu, \mu, \ldots,
\mu)$.

2) The unitary action of $\frk{k}_0$ on $\tilde{\mathcal H}$ can be
lifted to a unique unitary action of $K$ under which
\begin{eqnarray}
\tilde {\mathcal H} =\bigoplus_{l=0}^\infty \left(D(-l_\mu-1)\otimes
D^l\right)
\end{eqnarray} where $D^l$ is the irreducible module of $\mr{Spin}(2n+1)$
with highest weight $(l+\mu, \mu, \cdots, \mu)$ and $D(-l_\mu-1)$ is
the irreducible module of $\mr{Spin}(2)$ with weight $-l_\mu-1$.

3) $\tilde{\mathcal H}$ is a unitary $(\frk{g}, K)$-module.

4) $(\tilde\pi, \tilde{\mathcal H})$ is irreducible; in fact, it is
the unitary highest weight module of $\frk{g}$ with highest weight
$$\left(-(n+\mu-1/2),\mu, \cdots, \mu\right).$$
\end{Prop}
\begin{proof}
1) Recall that $\frk{s}$ is the $\frk{so}(2n+1)$ Lie sub algebra of
$\frk{g}$ generated by $$ \{H_i, E_{\pm e^j\pm e^k}, E_{\pm
e^l}\,|\, 1\le i \le n, 1\le j < k \le n, 1\le l\le n\}, $$  and
$\frk{s}_0:=\frk{s}\cap \frk{g}_0$ is the compact real form of
$\frk{s}$. Since $H_0$ commutes with any element in $\frk{s}$, in
view of Remark \ref{rmk2}, we conclude that each $\tilde{\ms H}_I$
is invariant under $\tilde \pi(\frk{s})$, i.e.,
$(\tilde\pi|_{\frk{s}}, \tilde{\ms H}_I)$ is a representation of
$\frk{s}$.

Inside $\frk{s}$ there is an $\frk{so}(2n)$ Lie sub algebra
$\frk{r}$. Note that $H_1$, \ldots, $H_n$ are the generators of a
Cartan subalgebra of $\frk{r}$, and $H_1$, \ldots, $H_{n+1}$ are the
generators of a Cartan subalgebra of $\frk{s}$. Recall from part 4)
of Proposition \ref{Prmk2},
\begin{eqnarray}\label{brch1}
(\tilde \pi|_{\frk{r}}, \tilde{\ms H}_I)=\left\{ \begin{array}{ll}\bigoplus_{l=0}^I
\left(\tilde D_l^+\oplus \tilde D_l^-\right) & \mbox{if $\mu=1/2$}\\
\\
\bigoplus_{l=0}^I \tilde D_l^0 & \mbox{if $\mu=0$}
\end{array}\right.
\end{eqnarray} where $\tilde D_l^\pm$ is the highest weight $\frk{r}$-module with
highest weight $(l+1/2, 1/2, \cdots, 1/2,\pm 1/2)$, and $\tilde
D_l^0$ is the highest weight $\frk{r}$-module with highest
weight $(l, 0, \cdots, 0)$.

By applying the branching rule\footnote{See, for example, Theorem 3
of \S 129 of Ref. \cite{DZ73}. } for $(\frk{s},\frk{r})$, one finds
that there is only one solution to Eq. (\ref{brch1}):
$(\tilde\pi|_{\frk{s}},\tilde{\ms H}_I)$ is the highest weight
module of $\frk{s}$ with highest weight $$(I+\mu,
\mu, \cdots, \mu).$$

Part 1) is done.

2) Since $\tilde{\ms H_l}$ is the space of square integrable
solutions of Eq. $ \hat\Gamma_{-1}\psi = (l_\mu+1)\psi$ and $\tilde
\pi(H_0)=-\hat\Gamma_{-1}$, as a $\frk{k}$-module, $\tilde{\ms
H_l}=D(-l_\mu-1)\otimes D^l$ where $D^l$ is the irreducible module
of $\mr{Spin}(2n+1)$ with highest weight $(l+\mu, \mu, \cdots, \mu)$
and $D(-l_\mu-1)$ is the irreducible module of $\mr{Spin}(2)$ with
weight $-l_\mu-1$. Since $\mu$ is a half integer, the irreducible
unitary action of $\frk{k}_0$ on $\tilde{\ms H_l}$ can be promoted
to a unique irreducible unitary action of $K$. Therefore,
$\tilde{\mathcal H}$ is a unitary $K$-module and has the following
decomposition into isotypic components of $K$:
$$
{\tilde{\mathcal H}}=\bigoplus_{l=0}^\infty \tilde{\ms
H}_l=\bigoplus_{l=0}^\infty \left(D(-l_\mu-1)\otimes D^l\right).
$$

3) From the definition, it is clear that the action of $K$ on ${\tilde{\mathcal H}}$ is
compatible with that of $\frk{g}$ on ${\tilde{\mathcal H}}$, and its linearization agrees with
the action of $\frk{k}_0$. Part 2) says that ${\tilde{\mathcal H}}$ is $K$-finite.
Therefore, ${\tilde{\mathcal H}}$ is a
unitary $(\frk{g}, K)$-module.

4) Let $v\neq 0$ be a vector in $\tilde{\ms H}_0$ with
$\frk{g}$-weight $\left(-(n+\mu-1/2),\mu, \cdots, \mu\right)$. Since
this weight is the highest among all weights with a nontrivial
weight vector in $\tilde{\mathcal H}$, $V:=U(\frk{g})\cdot v\subset
\tilde{\mathcal H}$ is the unitary highest weight $\frk{g}$-module
with highest weight $\left(-(n+\mu-1/2),\mu, \cdots, \mu\right)$.
Since $\tilde{\ms H}_l$ is irreducible under $\frk{s}\subset
\frk{g}$, either $\tilde{\ms H}_l\subset V$ or $\tilde{\ms H}_l\cap
V=0$, so in particular $\tilde{\ms H}_0\subset V$. We
\underline{claim} that $\tilde{\ms H}_l\subset V$ for any $l\ge 0$,
consequently $V=\tilde{\mathcal H}$ and then part 4) is done. To
prove the claim, we note that $U(\frk{sl}(2))\cdot v$ must be the
discrete series representation of $\frk{sl}(2)$ with highest
($H_0$-) weight $-(n+\mu-1/2)$ because it is a nontrivial unitary
highest weight representation of the non-compact Lie algebra
$\frk{sl}_0(2)$. In view of the fact that $\tilde{\ms H}_l$ is the
eigenspace of $\tilde\pi(H_0)$ with eigenvalue $-(l_\mu+1)$,
$\tilde{\ms H}_l\,\cap \,(U(\frk{sl}(2))\cdot
v)=\mr{span}\{\tilde\pi(E_-^l)(v)\}$ must be one-dimensional. Then
$\tilde{\ms H}_l\,\cap \,V\neq 0$ because $\dim (\tilde{\ms
H}_l\,\cap \,V)\ge \dim (\tilde{\ms H}_l\,\cap
\,(U(\frk{sl}(2))\cdot v))=1$.

\end{proof}
\subsection{Proof of Theorem \ref{main}}
Viewing the twisting map $\tau$ as an equivalence of
representations, we get a representation $\pi$ of $\frk{g}$ which is
equivalent to $\tilde \pi$. Then the two propositions proved in the
previous subsection are true if we drop all ``tilde" there. Thus
${\mathcal H}$ is the unitary highest weight $(\frk{g}, K)$-module
with highest weight $\left(-(n+\mu-1/2),\mu, \cdots, \mu\right)$. By
a standard theorem of Harish-Chandra\footnote{See, for example,
Theorem 7 on page 71 of Ref. \cite{BK96}}, we know that ${\ms H}$ is
the unitary highest weight $G$-module with highest weight
$\left(-(n+\mu-1/2),\mu, \cdots, \mu\right)$ such that $(\pi,
{\mathcal H})$ is the underlying $(\frk{g}, K)$-module. One can
check that this highest weight module occurs at the first reduction
point of the Enright-Howe-Wallach classification
diagram\footnote{Page 101, Ref. \cite{EHW82}. In our case
$z=A(\lambda_0)=n+1/2$. It is in Case II when $\mu=0$, and in Case
III when $\mu=1/2$, see LEMMA 11.3 on page 128, Ref. \cite{EHW82}.
Note that, while there are two reduction points when $\mu= 0$, there
is only one reduction point when $\mu=1/2$.}. So part 1) is done.
Part 2) of Theorem \ref{main} is just a consequence of part 3) of
Proposition \ref{prop1}, and part 3) of Theorem \ref{main} is just a
consequence of part 2) of Proposition \ref{prop2}.

\appendix
\section{Geometrically transparent description}
Assume $n\ge 2$ is an integer and $\mu=0$ or $1/2$. The purpose of
this appendix is to give a geometrically transparent description of
the unitary highest weight module of $\widetilde{\mr{Spin}}(2,
2n+1)$ with highest weight $\left(-(n+\mu-1/2),\mu, \cdots,
\mu\right)$.

As usual, let ${\mathcal S}^{2\mu}$ be the pullback bundle under the
natural retraction ${\bb R}_*^{2n}\to \mr{S}^{2n-1}$ of the vector
bundle $\mr{Spin}(2n)\times_{\mr{Spin}(2n-1)} {\bf s}^{2\mu}\to
\mr{S}^{2n-1}$ with the natural $\mr{Spin}(2n)$-invariant
connection. Let $d^Dx$ be the Lebesgue measure on ${\bb R}^{2n}$. As
is standard in geometry, we use $ L^2({\mathcal S}^{2\mu})$ to
denote the Hilbert space of square integrable (with respect to
$d^Dx$) sections of ${\mathcal S}^{2\mu}$. We have shown that
$\tilde{\ms H}(\mu)=L^2({\mathcal S}^{2\mu})$, therefore,
$\left(\tilde \pi, L^2({\mathcal S}^{2\mu})\right)$ is the unitary
highest weight module of $\widetilde{\mr{Spin}}(2, 2n+1)$ with
highest weight $\left(-(n+\mu-1/2),\mu, \cdots, \mu\right)$. To
describe the infinitesimal action of $\widetilde{\mr{Spin}}(2,
2n+1)$ on $C^\infty({\mathcal S}^{2\mu})$, it suffices to describe
how $M_{\alpha,0}$, $M_{D+1,0}$ and $M_{-1, 0}$ act as differential
operators. It is easy to see that $M_{\alpha,0}$, $M_{D+1,0}$ and
$M_{-1, 0}$ are equal to $i\sqrt r\nabla_\alpha\sqrt r$, ${1\over
2}\left(\sqrt r\Delta_\mu\sqrt r+r-{c\over r}\right)$ and ${1\over
2}\left(\sqrt r\Delta_\mu\sqrt r-r-{c\over r}\right)$ respectively.
Here $\Delta_\mu$ is the Laplace operator twisted by ${\mathcal
S}^{2\mu}$. For example, for $\psi\in C^\infty({\mathcal
S}^{2\mu})$, we have
\begin{eqnarray}
(M_{\alpha,0}\cdot \psi)(r, \Omega) & = & i\sqrt r\nabla_\alpha
\left (\sqrt r\psi(r, \Omega)\right).
\end{eqnarray}

\end{document}